\documentclass[11pt,a4paper]{article}

\usepackage{indentfirst,mathrsfs}
\usepackage{amsfonts,amsmath,amssymb,amsthm}
\usepackage{latexsym,amscd}
\usepackage{amsbsy}
\usepackage{color}
\usepackage[noadjust]{cite}
\usepackage{cite}
\usepackage{multirow}
\usepackage{makecell}
\usepackage{float}
\usepackage{cases}

\newtheorem{thm}{Theorem}
\newtheorem{lem}{Lemma}

\newtheorem{remark}{Remark}

\newcommand{\fnn}{{\mathbb F}_{2^{n}}}
\newcommand{\fn}{{\mathbb F}_{2^{2m}}}
\newcommand{\fm}{{\mathbb F}_{2^m}}

\newcommand{\bm}{\overline{b}}
\newcommand{\ta}{\overline{\theta}}

\begin{document}
\title{On the Niho type locally-APN power functions and their boomerang spectrum}
\author{Xi Xie, Sihem Mesnager, Nian Li, Debiao He, Xiangyong Zeng
\thanks{X. Xie, N. Li, and X. Zeng are with the Hubei Key Laboratory of Applied Mathematics, Faculty of Mathematics and
Statistics, Hubei University, Wuhan, China.
S. Mesnager is with the Department of Mathematics, University of Paris VIII, F-93526 Saint-Denis, Paris, France, the Laboratory Geometry, Analysis and Applications, LAGA, University Sorbonne Paris Nord CNRS,
UMR 7539, F-93430, Villetaneuse, France, and also with the Telecom Paris, 91120 Palaiseau, France.
D. He is with the School of Cyber Science and Engineering, Wuhan University,
Wuhan, China.
Email: xi.xie@aliyun.com,  smesnager@univ-paris8.fr, nian.li@hubu.edu.cn, hedebiao@whu.edu.cn, xzeng@hubu.edu.cn.}
}
\date{}
\maketitle
\begin{quote}
{{\bf Abstract:}

In this article, we focus on the concept of locally-APN-ness (``APN" is  the abbreviation of  the well-known notion of Almost Perfect Nonlinear)  introduced by Blondeau, Canteaut, and Charpin, which makes the corpus of S-boxes somehow larger regarding their differential uniformity and, therefore, possibly, more suitable candidates against the differential attack (or their variants). Specifically, given  two coprime positive integers $m$ and $k$ such that $\gcd(2^m+1,2^k+1)=1$, we investigate the locally-APN-ness property of an infinite family of Niho type power functions in the form $F(x)=x^{s(2^m-1)+1}$ over the finite field $\fn$ for $s=(2^k+1)^{-1}$, where $(2^k+1)^{-1}$ denotes the multiplicative inverse modulo $2^m+1$.

By employing finer studies of the number of solutions of certain equations over finite fields (with even characteristic) as well as some subtle manipulations of solving some equations, we prove that $F(x)$ is locally APN and determine its differential spectrum. It is worth noting that computer experiments show that this class of locally-APN power functions covers all Niho type locally-APN power functions for $2\leq m\leq10$. In addition, we also determine the boomerang spectrum of $F(x)$ by using its differential spectrum, which particularly generalizes a recent result by Yan, Zhang, and Li.
}

\end{quote}

\section{Introduction}

A substitution box (S-box) in a block cipher is a mapping that takes $n$ binary inputs and whose image is a binary $m$-tuple for some positive integers $n$ and $m$, which is usually the only nonlinear element of the most modern block ciphers. Therefore, it is significant to employ S-boxes with good cryptographic properties to resist various attacks.

Differential attack, introduced by Biham and Shamir \cite{BS}, is one of the most efficient attacks on a block cipher. For an $n$-bit S-box $F(x)$, i.e., a mapping from the finite field $\fnn$ to $\fnn$, the properties for differential propagations of $F$ are captured in the {\rm DDT} (Difference Distribution Table) of $F$ which are given by
$${\rm DDT}_{F}(a,\,b)=|\{x\in\fnn: F(x)+F(x+a)=b\}|$$
for any $a,\,b\in\fnn$.
The differential uniformity of $F$ is defined as
$$\delta(F)=\max_{a,b\in\mathbb{F}_{2^n},\,a\ne0}{\rm DDT}_{F}(a,\,b).$$
Differential uniformity is an important concept in cryptography introduced by Nyberg \cite{Nyberg-93,SM-NK} as it quantifies the degree of security of the cipher concerning the differential attack if $F$ is involved as an S-box in such cipher.
It is easy to see that $\delta(F)$ is always even and $\delta(F)\geq2$.
The function $F$ with $\delta(F)=2$ is called Almost Perfect Nonlinear (APN) and offers maximal resistance to differential attacks.
The differential spectrum of $F$, defined as the multi-set $\{{\rm DDT}_{F}(a,\,b): a\in\fnn^*,\,b\in\fnn\}$,
may influence its security regarding some variants of differential cryptanalysis \cite{BCC}. Those functions are very important in symmetric cryptography since they contribute to an optimal resistance against differential cryptanalysis, a powerful attack employed against block ciphers. A lot of attention and efforts have been made, as can be seen notably in the nice and complete Chapter 11 in the recent book \cite{Carlet-book-2021}.

Another important cryptanalysis technique on block ciphers is the boomerang attack, introduced by Wagner in \cite{W}, which can be considered an extension of the classical differential attack \cite{BS}. Analogous to the ${\rm DDT}$ concerning the differential attack, Cid et al. \cite{CHPSS} introduced a tool called Boomerang Connectivity Table (BCT) to analyze the boomerang attack of block ciphers. Let $F: \fnn\rightarrow \fnn$ be a permutation. The entries of the {\rm BCT} of $F$ are given by
$${\rm BCT}_{F}(a,\,b)=\big|\big\{x\in\fnn : F^{-1}(F(x)+b)+F^{-1}(F(x+a)+b)=a\big\}\big|$$
for any $a,\,b\in\fnn$. Then the boomerang uniformity of $F$ defined as
$$\beta(F)=\max_{a,b\in\mathbb{F}_{2^n}^*}{\rm BCT}_{F}(a,\,b),$$
was introduced by Boura and Canteaut \cite{BC} to quantify the resistance of $F$ against the boomerang attack. It was known in \cite{CHPSS} that $\beta(F)\geq\delta(F)$, and if $\delta(F)=2$, then $\beta(F)=2$, hence APN permutations offer maximal resistance to both differential and boomerang attacks. Similarly, the boomerang spectrum of $F$, defined as the multi-set $\{{\rm BCT}_{F}(a,\,b): a,b\in\fnn^*\}$, was introduced to estimate the resistance of a block cipher against some variants of boomerang cryptanalysis \cite{JLLQ}. The reader can refer to the following non-exhaustive list of references \cite{BCC,BCC1,BP,BC,EM,HPS2,KMCLLJ,MTX} for some results on the permutations with low differential and boomerang uniformities. A recent survey on some generalizations of the differential and boomerang uniformities is presented in \cite{MMM}. APN permutations offer maximal resistance to differential and boomerang attacks, but there are extremely difficult to construct despite many efforts and recent advances. The high importance of such families of permutations for real applications comes from the difficulty of finding APN permutations in even dimension (known as the Big APN Problem \cite{BDMW}), making this topic still exciting nowadays.

Power functions are preferred candidates for S-boxes since they have simple algebraic forms and usually have a lower implementation cost in hardware. Their particular algebraic structure makes the determination of their differential properties easier to handle. For a power function $F(x)=x^d$ over $\fnn$, where $1\le d\le 2^n-2$ is a positive integer.  One  can see that ${\rm DDT}_F(a,\,b)={\rm DDT}_F(1,\,b/a^d)$ for all $a\in\mathbb{F}_{2^n}^*,\,b\in\mathbb{F}_{2^n}$ and ${\rm BCT}_F(a,\,b)={\rm BCT}_F(1,\,b/a^d)$ for all $a,\,b\in\mathbb{F}_{2^n}^*$.
That is to say, the differential (resp. boomerang) spectrum of $F(x)$ is completely determined by the values of ${\rm DDT}_F(1,\,b)$ (resp. ${\rm BCT}_F(1,\,b)$) as $b$ runs through $\mathbb{F}_{2^n}$ (resp. $\mathbb{F}_{2^n}^*$). Therefore, the differential spectrum of a power function $F(x)$ with differential uniformity $\delta(F)$ is defined as
$$\mathbb{DS}_F=\{\omega_i>0:0\leq i\leq\delta(F)\},$$
where $\omega_i$ denotes the number of output differences $b\in\fnn$ that occur $i$ times, that is, $\omega_i=|\{b\in\fnn: {\rm DDT}_F(1,b)=i\}|$ for $0\leq i\leq\delta(F)$.
Similarly, the boomerang spectrum of a power function $F(x)$ with boomerang uniformity $\beta(F)$ is simply defined as
$$\mathbb{BS}_F=\{\nu_i>0:0\leq i\leq\beta(F)\},$$
where $\nu_i=|\{b\in\fnn^*: {\rm BCT}_F(1,b)=i\}|$ for $0\leq i\leq\beta(F)$.
As a generalization of the APN-ness  property, Blondeau, Canteaut, and  Charpin \cite{BCC1} introduced the notion of locally APN power functions. A power function $F(x)$ over $\fnn$ is said to be locally-APN if $\max \{{\rm DDT}_F(1,b): b\in \fnn \backslash \mathbb{F}_{2}\} =2$. To our knowledge, very few locally APN power functions are known in the literature.
Typically, it is also hard to determine the differential and boomerang spectrum for a given function. The infinite families of power functions with known differential and boomerang spectrums are listed in Table \ref{differential-table} and Table \ref{boomerang-table}, respectively.

\begin{table}[!htb]\footnotesize
\caption{The power function $F(x)=x^d$ over $\fnn$ for which its differential spectrum is known}  \label{differential-table}
\renewcommand\arraystretch{0.8}
\setlength\tabcolsep{2pt}
\centering
\begin{tabular}{|c|c|c|c|c|}
\hline No.&  $d$                & Condition                   &  $\delta(F)$    & Refs.                  \\ \hline\hline
  1  &  $2^t+1$            &    $\gcd(t,n)=s$            &    $2^s$       &\cite{BCC,EM}            \\ \hline
  2  &  $2^{2t}-2^t+1$     & $\gcd(t,n)=s$, $n/s$ odd    &    $2^s$        &\cite{BCC}            \\ \hline
  3  &  $2^n-2$            &$n\geq 2$                    & $2$ or $4$ (locally-APN)     &\cite{BCC,EM}            \\ \hline
  4  & $2^{2k}+2^k+1$      &$n=4k$                       & 4               &\cite{BCC,XY}       \\ \hline
  5  & $2^t-1$             &$t=3,n-2$                    & 6 or 8          &\cite{BCC1}             \\ \hline
  6  &$2^t-1$             &$t=n/2, n/2+1$, $n$ even     & $2^{n/2}-2$ or $2^{n/2}$ (locally-APN)  &\cite{BCC}   \\ \hline
  7  & $2^t-1$             &$t=(n-1)/2,(n+3)/2$, $n$ odd & $6$ or $8$      &\cite{BP}             \\ \hline
  8  & $2^{3k}+2^{2k}+2^{k}-1$ & $n=4k$                  & $2^{2k}$        &\cite{LWZT}             \\ \hline
  9  & $2^m+2^{(m+1)/2}+1$ &$n=2m$, $m\geq5$ odd         & $8$             &\cite{XYY}             \\ \hline
  10 & $2^{m+1}+3$         &$n=2m$, $m\geq5$ odd         & $8$             &\cite{XYY}             \\\hline
  11 & $k(2^m-1)$         & $n=2m$, $\gcd(k,2^m+1)=1$        &$2^m-2$ (locally-APN)              &   \cite{HLXZT}           \\\hline
  12 & $(2^m-1)/(2^k+1)+1$        & $n=2m$, $\gcd(k,m)=1$  &$2^m$ (locally-APN)   & This paper \\\hline
\end{tabular}
\end{table}

\begin{table}[!htb]\footnotesize
\caption{Power functions $F(x)=x^d$ over $\fnn$ with known boomerang spectrum}  \label{boomerang-table}
\renewcommand\arraystretch{0.8}
\setlength\tabcolsep{2pt}
\centering
\begin{tabular}{|c|c|c|c|c|c|}
\hline No.&  $d$      & Condition                               & $\beta(F)$    & Refs.                \\ \hline\hline
  1   & $2^n-2$       & $n\geq 2$                                 & $4$ or $6$    &\cite{BC,EM,JLLQ}           \\ \hline
  2  &  $2^t+1$       &$s=\gcd(t,n)$                            & $2^s$ or $2^s(2^s-1)$   &\cite{BC,EM,HPS2} \\ \hline
  3  &  $2^{m+1}-1$  &$n=2m$, $m>1$                             & $2^m+2$       &\cite{YZZ}   \\\hline
  4  &  $k(2^m-1)$  & $n=2m$, $\gcd(k,2^m+1)=1$        & $2$ or $4$       & \cite{HLXZT}  \\\hline
  5  &  $(2^m-1)/(2^k+1)+1$  &$n=2m$, $\gcd(k,m)=1$                             & $2^m+2$       &This paper   \\\hline
\end{tabular}
\end{table}

In this paper, we study the locally-APN-ness property of the Niho type power functions over $\fnn$, namely, the power functions of the form $F(x)=x^{s(2^m-1)+1}$, where $n=2m$ and $1\leq s\leq 2^m$. Concretely, for a positive integer $k$ satisfying $\gcd(k,\,m)=1$ and $\gcd(2^k+1,\,2^m+1)=1$, we prove that $F(x)=x^{s(2^m-1)+1}$ with $s=(2^k+1)^{-1}$ is a locally-APN function over $\fn$, where $(2^k+1)^{-1}$ denotes the multiplicative inverse modulo $2^m+1$. Our computer experiments show that this class of locally-APN power functions covers all Niho type locally-APN power functions for $2\leq m\leq10$. Moreover, by carrying out some finer manipulations of solving certain equations over finite fields, we completely determine both the differential spectrum of $F(x)$ and the boomerang spectrum of $F(x)$, which generalizes the results in \cite{BCC} and \cite{YZZ} respectively.

The rest of this paper is organized as follows. Section \ref{prel} introduces the preliminaries. Section \ref{diff} ultimately determines the differential spectrum of $F(x)=x^{s(2^m-1)+1}$ with $s=(2^k+1)^{-1}$,  leading to the fact that $F$ is actually  locally-APN. Section \ref{boom} determines the
 boomerang spectrum of $F(x)$ based on its differential spectrum. Section \ref{conc} concludes this study.

\section{Preliminaries}\label{prel}
Throughout this paper, $|E|$ denotes the cardinality of a finite set $E$. In addition, let  $n$ be a positive integer and $\mathbb F_{2^n}$ be the finite field (of characteristic $2$) of order $2^n$. We denote by $\mathbb F_{2^n}^*$ the multiplicative cyclic group of non-zero elements of $\mathbb F_{2^n}$.  The (absolute) trace function ${\rm Tr}_1^n:\mathbb F_{2^n}\longrightarrow \mathbb F_2$ is defined by ${\rm Tr}_1^n(x)=\sum_{i=0}^{n-1} x^{2^{i}}$ for all $x\in\mathbb F_{2^n}$.  A positive integer $d$ is called a Niho exponent with respect to the finite field $\fn$ if $d\equiv 2^i \,({\rm mod}\, 2^m-1)$ for some $i<2m$. When $i= 0$, the integer $d$ is then called a normalized Niho exponent. For simplicity, denote the conjugate of $x\in\fn$ over $\fm$ by $\overline{x}$, i.e., $\overline{x}=x^{2^m}$. The unit circle of $\fn$ is defined as follows:
$$\mu_{2^m+1}:=\{v\in\fn: v\overline{v}=1\}.$$
Note that $\mu_{2^m+1}\cap\fm=\{1\}$. It is well-known that each $x\in\fn^*$ can be uniquely written as $x=uv$ for some $u\in\fm^*$ and $v\in\mu_{2^m+1}$.

The following lemmas are helpful for the subsequent sections.

\begin{lem}{\rm(\cite{R})}\label{lem.u}
For a positive integer $m$ and any element $\theta\in\fn\backslash\fm$, the mapping $\varphi: x \rightarrow (x+\overline{\theta})/(x+\theta)$ from $\fm$ to $\mu_{2^m+1}\backslash\{1\}$ is a bijection.
\end{lem}

\begin{lem}{\rm(\cite{TZLH})}\label{lem.eq.x^2}
Let $n=2m$ be an even positive integer and $a,b\in\mathbb{F}_{2^n}^*$. Then
$x^2+ax+b=0$ has two solutions in $\mu_{2^m+1}$ if and only if $b=a^{1-2^m}$ and ${\rm Tr}_1^m(b/a^2)=1$. \end{lem}

\begin{lem}{\rm(\cite{MKCLG})}\label{lem.eq1}
Let $n,\,r$ be positive integers such that $\gcd(n,\,r)=1$. For any $a\in\mathbb{F}_{2^n}$, the equation
$x^{2^r}+x=a$
has either $0$ or $2$ solutions in $\mathbb{F}_{2^n}$. Moreover, it is solvable with two solutions in $\mathbb{F}_{2^n}$ if and only if ${\rm Tr}_1^n(a)=0$.
\end{lem}

{\begin{lem}{\rm(\cite{WW})}\label{lem.theta}
Let $n,\,k$ be positive integers and $x,\,y$ be two elements of $\fnn$. Then
$$x^{2^k+1}+y^{2^k+1}=(x+y)^{2^k+1}+\sum_{i=0}^{k-1}(xy)^{2^i}(x+y)^{2^k-2^{i+1}+1}.$$
\end{lem}

According to Lemma 22 and the proof of Theorem 23 in \cite{DFHR}, we can obtain the following result, which will play a significant role in proving our main result.

\begin{lem}{\rm(\cite{DFHR})}\label{lem.eq}
Let $n$ and $r$ be positive integers with $r_0=\gcd(r,\,n)$. Then
$$Q(x)=x^{2^r+1}+ax^{2^r}+bx+c,\,a,b,c\in\mathbb{F}_{2^n}$$
has either $0,\,1,\,2$ or $2^{r_0}+1$ roots in $\mathbb{F}_{2^n}$.
Specially, if $n=2m$ and $x_0,\,x_1,\,x_2\in\mu_{2^m+1}$ are three distinct roots of $Q$, then
$Q$ has $2^{r_1}+1$ distinct roots in $\mu_{2^m+1}$, where $r_1=\gcd(r_0,\,m)$.
Further, $x_0+Ax_1+(1+A)x_2\ne0$ for any $A\in\mathbb{F}_{2^{r_1}}$ and each root of $Q$ in $\mu_{2^m+1}\backslash\{x_0,\,x_1,\,x_2\}$ can be parameterized as
$$x_A=\frac{x_1x_2+Ax_0x_2+(1+A)x_0x_1}{x_0+Ax_1+(1+A)x_2}$$
for some $A\in\mathbb{F}_{2^{r_1}}\backslash \mathbb{F}_{2}$.
\end{lem}

\section{The differential spectrum of the Niho power function $F(x)=x^{s(2^m-1)+1}$ }\label{diff}

In this section, we determine the differential spectrum of  $F(x)=x^{s(2^m-1)+1}$ with $s=(2^k+1)^{-1}$ over $\fn$, where $m$ and $k$ are positive integers such that $\gcd(2^k+1,\,2^m+1)=1$ and $\gcd(k,\,m)=1$.

To determine the differential spectrum of $F(x)$, we first transfer the problem of determining the number of solutions of $F(x+1)+F(x)=b$ for any $b\in\fn$ to that of finding the number of solutions to a system of equations by using the polar representations of $x+1$ and $x$. This enables us to reduce the problem to investigate the number of solutions to a system of quadratic equations because $F(x)$ is of Niho type and $s=(2^k+1)^{-1}$. Our first main result is then obtained via a thorough discussion on whether the solutions of a four-term quadratic equation are in the unit cycle of $\fn$ or not.

\begin{thm}\label{thm.DS}
Let $m$ and $k$ be positive integers with $\gcd(2^k+1,\,2^m+1)=1$ and $s=(2^k+1)^{-1}$ denotes the multiplicative inverse modulo $2^m+1$. If $\gcd(k,\,m)=1$, then the power function $F(x)=x^{s(2^m-1)+1}$ over $\fn$ is locally-APN and its differential spectrum is given by
$$\mathbb{DS}_F=\{\omega_0=2^{2m-1}+2^{m-1}-1,\,\omega_2=2^{2m-1}-2^{m-1},\,\omega_{2^m}=1\}.$$
\end{thm}

\begin{proof}
It suffices to calculate the number of solutions in $\fn$ of
\begin{equation}\label{D(F)}
(x+1)^{s(2^m-1)+1}+x^{s(2^m-1)+1}=b
\end{equation}
for any $b\in\fn$. We shall distinguish the cases $b=1$ and $b\ne1$ as follows.

\textbf{Case 1}: $b=1$.
If \eqref{D(F)} holds, i.e., we have
\begin{equation}\label{eq1}
(x+1)^{s(2^m-1)+1}=x^{s(2^m-1)+1}+1.
\end{equation}
Taking $2^m$-th power on both sides of \eqref{eq1} gives
\begin{equation}\label{eq2}
(x+1)^{-s(2^m-1)+2^m}=x^{-s(2^m-1)+2^m}+1.
\end{equation}
Multiplying \eqref{eq1} and \eqref{eq2} yields
$$(x+1)^{2^m+1}=x^{2^m+1}+x^{s(2^m-1)+1}+x^{-s(2^m-1)+2^m}+1,$$
which can be simplified as
$$(x^{-s(2^m-1)}+1)(x^{s(2^m-1)+1}+x^{2^m})=0.$$
Then we have either $x^{-s(2^m-1)}=1$ or $x^{2^m}(x^{(s-1)(2^m-1)}+1)=0$, which implies that $x\in\fm$ since $\gcd(s,\,2^m+1)=1$ and $\gcd(s-1,\,2^m+1)=1$ due to  $\gcd(2^k+1,\,2^m+1)=1$, $s=(2^k+1)^{-1}$ and $s-1=-2^k(2^k+1)^{-1}$. On the other hand, it can be readily verified that  $F(x+1)+F(x)=1$ if $x\in\fm$.
This proves that $F(x+1)+F(x)=1$ if and only if $x\in\fm$, i.e., ${\rm DDT}_F(1,\,1)=2^m$.

\textbf{Case 2}: $b\ne1$. Clearly, we have $x\in\fn\backslash\fm$ in this case. Let $x=u_1v_1$ and $x+1=u_2v_2$, where $u_1,\,u_2\in\fm^*$ and $v_1,\,v_2\in\mu_{2^m+1}\backslash \{1\}$. Then \eqref{D(F)} is equivalent to the system of equations
\begin{equation}\label{sys.uivi}
\begin{aligned}
  u_1v_1^{1-2s}+u_2v_2^{1-2s} &= b, \\
  u_1v_1+u_2v_2 &= 1.
\end{aligned}
\end{equation}
Eliminating the term $u_2$ from \eqref{sys.uivi} results in
$(1+v_1^{-2s}v_2^{2s})u_1v_1 = 1+bv_2^{2s},$
 i.e.,
\begin{equation}\label{eq.v1}
(v_1^{2s}+v_2^{2s})u_1v_1 = (1+bv_2^{2s})v_1^{2s}.
\end{equation}
Similarly, by eliminating the term $u_1$ from \eqref{sys.uivi}, we can obtain
\begin{equation}\label{eq.v2}
(v_1^{2s}+v_2^{2s})u_2v_2 = (1+bv_1^{2s})v_2^{2s}.
\end{equation}
We claim that $v_1^{2s}\ne v_2^{2s}$. Suppose that $v_1^{2s}=v_2^{2s}$. Then we have $v_1=v_2$ due to $\gcd(2s,\,2^m+1)=1$. Since $u_1v_1+u_2v_2=1$ and $\fm^*\bigcap\mu_{2^m+1}=\{1\}$, we have $u_1+u_2=v_1^{-1}=1$, which contradicts to $v_1\ne 1$. Hence $v_1^{2s}+v_2^{2s}\ne0$, i.e., $v_1\ne v_2$. Then by \eqref{eq.v1} and \eqref{eq.v2}, one gets
\begin{equation}\label{u1.u2}
u_1=\frac{(1+bv_2^{2s})v_1^{2s}}{(v_1^{2s}+v_2^{2s})v_1},\,\,\,
u_2=\frac{(1+bv_1^{2s})v_2^{2s}}{(v_1^{2s}+v_2^{2s})v_2},
\end{equation}
and
\begin{eqnarray*}
  u_1^{2^m}&=&\frac{(1+\bm v_2^{-2s})v_1^{-2s}}{(v_1^{-2s}+v_2^{-2s})v_1^{-1}}=\frac{(v_2^{2s}+\bm)v_1}{v_1^{2s}+v_2^{2s}}, \\
  u_2^{2^m}&=& \frac{(1+\bm v_1^{-2s})v_2^{-2s}}{(v_1^{-2s}+v_2^{-2s})v_2^{-1}}=\frac{(v_1^{2s}+\bm)v_2}{v_1^{2s}+v_2^{2s}}
\end{eqnarray*}
since $v_i^{2^m}=v_i^{-1}$ for $i=1,2$, where $\bm=b^{2^m}$.
Due to the fact $u_1,\,u_2\in\fm^*$, we then have
\begin{equation}\label{sys.v1v2}
\begin{aligned}
bv_1^{2s}v_2^{2s}+ v_1^{2s}+v_1^{2}v_2^{2s}+\bm v_1^{2}&= 0, \\
  bv_1^{2s}v_2^{2s}+ v_1^{2s}v_2^2+v_2^{2s}+\bm v_2^{2}&= 0.
\end{aligned}
\end{equation}
From \eqref{u1.u2} one can see that $u_1$ and $u_2$ can be determined by $v_1$ and $v_2$. Therefore we only need to determine the pairs $(v_1,\,v_2)$ with distinct $v_1,\,v_2\in\mu_{2^m+1}\backslash \{1\}$ such that \eqref{sys.v1v2} holds.

Again by $\gcd(2s,\,2^m+1)=1$, one can conclude that there exist unique $y\in\mu_{2^m+1}\backslash \{1\}$ and $z\in\mu_{2^m+1}\backslash \{1\}$ such that $v_1^{2s}=y$ and $v_2^{2s}=z$ respectively. Clearly, we have $y\ne z$ since $v_1\ne v_2$. Then \eqref{sys.v1v2} can be rewritten as
\begin{equation}\label{sys.z}
\begin{aligned}
(z+\bm) y^{2^k+1}+(bz+ 1)y&= 0, \\
(y+\bm) z^{2^k+1}+(by+ 1)z&= 0.
\end{aligned}
\end{equation}
Note that $yz\ne0$. Then \eqref{sys.z} is equivalent to
\begin{eqnarray}
  (y^{2^k}+b)z+\bm y^{2^k}+ 1&=& 0, \label{z-eq1}\\
  (y+\bm) z^{2^k}+by+ 1&=& 0. \label{z-eq2}
\end{eqnarray}
We claim that $y^{2^k}+b\ne0$. Otherwise, substituting $b=y^{2^k}$ into \eqref{z-eq2} yields
\begin{equation}\label{z-eq2.1}
(y+y^{-2^k})z^{2^k}=y^{2^k+1}+1.
\end{equation}
Since $\gcd(2^k+1,\,2^m+1)=1$ and $y\in\mu_{2^m+1}\backslash \{1\}$, we have $y^{2^k+1}\ne1$, which implies $y+y^{-2^k}=y^{-2^k}(y^{2^k+1}+1)\ne0$. Dividing both sides of \eqref{z-eq2.1} by $y+y^{-2^k}$ gives $z^{2^k}=y^{2^k}$, a contradiction with $y\ne z$. Hence $y^{2^k}+b\ne0$.
Then \eqref{z-eq1} indicates $z=(\bm y^{2^k}+ 1)/(y^{2^k}+b)$. Substituting it into \eqref{z-eq2}, we can obtain
\begin{equation}\label{y-eq}
(y+\bm)(\bm y^{2^k}+ 1)^{2^k}+(by+ 1)(y^{2^k}+b)^{2^k}=0,
\end{equation}
that is,
\begin{equation}\label{y-eq1}
(b+\bm^{2^k})y^{2^{2k}+1}+(\bm^{2^k+1}+1) y^{2^{2k}}+ (b^{2^k+1}+1)y+b^{2^k}+\bm= 0.
\end{equation}
Recall that $x=u_1v_1$, $u_1=[(1+bv_2^{2s})v_1^{2s}]/[(v_1^{2s}+v_2^{2s})v_1]$, $v_1^{2s}=y$, $v_2^{2s}=z$ and $z=(\bm y^{2^k}+ 1)/(y^{2^k}+b)$.
This indicates that $x$ is uniquely determined by $y$. Therefore, if we define
\begin{eqnarray}\label{eq-phi}
\Phi=\{y\in\mu_{2^m+1}: y \ne 1, y^{2^k}+b\ne0, y\ne \frac{\bm y^{2^k}+ 1}{y^{2^k}+b} {\rm \;and}\; \eqref{y-eq1} {\rm \;holds}\},
\end{eqnarray}
we then immediately have ${\rm DDT}_F(1,\,b)\leq|\Phi|$ when $b\ne 1$.

To determine the cardinality of $\Phi$, we first show that $z=(\bm y^{2^k}+ 1)/(y^{2^k}+b)$ is also in $\Phi$ if $y\in\Phi$. Note that $y+\bm\ne0$. Otherwise, we have $b\in\mu_{2^m+1}\backslash\{1\}$ and $z=\bm(y^{2^k}+ b)/(y^{2^k}+b)=\bm=y$, a contradiction. Then, by  \eqref{y-eq}, one gets
$z^{2^k}=( by+ 1)/(y+\bm)$. Thus, $z$ satisfies
\begin{equation}\label{eq-z}
z=\frac{\bm y^{2^k}+ 1}{y^{2^k}+b}, \;\;\; z^{2^k}=\frac{by+ 1}{y+\bm}.
\end{equation}
By \eqref{eq-phi}, to prove $z\in\Phi$, it suffices to prove 1) $z\in\mu_{2^m+1}\backslash\{1\}$; 2) $z^{2^k}+b\ne0$; 3) $z\ne (\bm z^{2^k}+ 1)/(z^{2^k}+b)$; and 4) $z$ satisfies \eqref{y-eq1}. A simple calculation gives $z\in\mu_{2^m+1}$ due to $y\in\mu_{2^m+1}$. If $z=1$, then by \eqref{eq-z} and the fact $b\ne 1$, one gets $y^{2^k}=(b+1)/(\bm+1)$ and $y=(\bm+1)/(b+1)$. This leads $y^{2^k+1}=1$ and then $y=1$ since $\gcd(2^k+1,2^m+1)=1$ and $y\in\mu_{2^m+1}$, a contradiction with $y\ne1$. Thus $z\ne 1$. This proves 1); If $z^{2^k}+b=0$, then $b\in\mu_{2^m+1}$ and further \eqref{eq-z} gives $z=\bm$ and $z^{2^k}=b$, which implies $z^{2^k+1}=1$, i.e., $z=1$ due to $\gcd(2^k+1,2^m+1)=1$ and $z\in\mu_{2^m+1}$, a contradiction with $z\ne1$. Thus $z^{2^k}+b\ne0$. This proves 2) and we can obtain $z+\bm\ne0$ in the same manner; Observe that \eqref{eq-z} also gives
\[ y^{2^k}=\frac{bz+ 1}{z+\bm}, \;\;\;y=\frac{\bm z^{2^k}+ 1}{z^{2^k}+b}\]
due to $z+\bm\ne 0$ and $z^{2^k}+b\ne0$.
This together with \eqref{eq-z} indicates that $z\ne(\bm z^{2^k}+ 1)/(z^{2^k}+b)$ since $y\ne(\bm y^{2^k}+ 1)/(y^{2^k}+b)$. Moreover, one obtains $[(\bm z^{2^k}+ 1)/(z^{2^k}+b)]^{2^k}=(bz+ 1)/(z+\bm)$, which implies that $z$ satisfies \eqref{y-eq1}. This proves 3) and 4). Hence, $z=(\bm y^{2^k}+ 1)/(y^{2^k}+b)\in\Phi$ if $y\in\Phi$.
This shows that $|\Phi|$ is even since $y\ne (\bm y^{2^k}+ 1)/(y^{2^k}+b)$.

 Next, we consider the value of $|\Phi|$ as follows:

\textbf{Case 2.1}: $b=0$. If $b=0$, then \eqref{y-eq1} is reduced to $y^{2^{2k}}+ y=0$, i.e., $y^{2^{2k}-1}=1$. Consequently, one gets $y^{\gcd(2^{2k}-1,\,2^m+1)}=y^{\gcd(2^{k}-1,\,2^m+1)}=1$ due to $y\in\mu_{2^m+1}$ and $\gcd(2^k+1,\,2^m+1)=1$.
Note that $\gcd(2^{k}-1,\,2^m+1)=1$ if $k$ is odd and $\gcd(2^{k}-1,\,2^m+1)=3$ if $k$ is even since $\gcd(k,\,m)=1$. Hence, \eqref{y-eq1} has $0$ or $2$ roots in $\mu_{2^m+1}\backslash \{1\}$, which implies $|\Phi|\leq2$ and ${\rm DDT}_F(1,\,0)\leq|\Phi|\leq2$.

\textbf{Case 2.2}: $b\ne0$.  Note that $\gcd(2^k+1, 2^m+1)=1$ if and only if one of $m/\gcd(k,\,m)$ and $k/\gcd(k,\,m)$ is even. This together with $\gcd(k,\,m)=1$, one has that $m$ and $k$ have different parity. Further, one can obtain $\gcd(2^{m+k}-1,2^{2m}-1)=1$ due to $\gcd(m+k,\,2m)=1$. This leads to $b+\bm^{2^k}\ne0$ if $b\ne0,\,1$.
Then by Lemma \ref{lem.eq}, \eqref{y-eq1} has either $0,\,1,\,2$ or $5$ solutions in $\fn$ since $\gcd(2k,2m)=2$.

In what follows, we prove that $|\Phi|\leq2$ for all $b\in\fn\backslash\mathbb{F}_2$.

Clearly, we have $|\Phi|\leq2$ if \eqref{y-eq1} has at most two distinct solutions in $\mu_{2^m+1}$. Now suppose that $y_0,\,y_1,\,y_2$ are three distinct solutions of \eqref{y-eq1} in $\mu_{2^m+1}$. Again by Lemma \ref{lem.eq}, we have that \eqref{y-eq1} has $2^{r_1}+1$ solutions in $\mu_{2^m+1}$, where $r_1=\gcd(\gcd(2k,2m),\,m)=\gcd(2,\,m)$.
If $m$ is odd, then \eqref{y-eq1} has $3$ solutions in $\mu_{2^m+1}$, which implies $|\Phi|\leq2$ since $|\Phi|$ is even.
If $m$ is even, then \eqref{y-eq1} has $5$ solutions in $\mu_{2^m+1}$ and consequently $|\Phi|\in\{0,2,4\}$. Thus, we only need to prove $|\Phi|\ne 4$ when $m$ is even.

Suppose that $|\Phi|=4$. Then by the fact that $(\bm y^{2^k}+ 1)/(y^{2^k}+b)\in \Phi$ if $y\in \Phi$, we can assume
\begin{equation}\label{Phi1}
\Phi=\{{y_1,y_2=\frac{\bm y_1^{2^k}+ 1}{y_1^{2^k}+b},y_3,y_4=\frac{\bm y_3^{2^k}+ 1}{y_3^{2^k}+b}}\}.
\end{equation}
Let $y_0\in\mu_{2^m+1}\backslash\Phi$ be the fifth solution of \eqref{y-eq1}. Then, according to \eqref{eq-phi}, $y_0$ must satisfy at least one of the following cases: 1) $y_0=1$;  2) $y_0^{2^k}=b$; 3) $y_0=(\bm y_0^{2^k}+ 1)/(y_0^{2^k}+b)$.
We first show that case 2) cannot occur. If $y_0^{2^k}=b$, then $b\in\mu_{2^m+1}$, which implies $b^{2^k+1}+1\ne 0$ due to $\gcd(2^k+1,2^m+1)=1$ and $b\ne1$. Further, \eqref{y-eq1} turns into
\[(b+b^{-2^k})y^{2^{2k}+1}+(b^{-(2^k+1)}+1) y^{2^{2k}}+ (b^{2^k+1}+1)y+b^{2^k}+b^{-1}= 0,\]
which can be factored into
\[b^{-(2^k+1)}(b^{2^k+1}+1)(by+1)(y^{2^k}+b)^{2^k}=0.\]
This leads to $by+1=0$ or $y^{2^k}+b=0$, which indicates that \eqref{y-eq1} has at most two solutions, a contradiction with the fact that \eqref{y-eq1} has $5$ solutions in $\mu_{2^m+1}$. Next we demonstrate that case 1) is covered by case 3). If $y_0=1$, then by \eqref{y-eq1} one gets $(b+1)^{2^k+1}+(\bm+1)^{2^k+1}= 0$ which implies $(b+1)/(\bm+1)=1$ since $b\ne1$, $(b+1)/(\bm+1)\in\mu_{2^m+1}$ and $\gcd(2^k+1,2^m+1)=1$, i.e., $b=\bm$. Consequently, $1=y_0=(\bm y_0^{2^k}+ 1)/(y_0^{2^k}+b)$. Therefore, if $y_0\in\mu_{2^m+1}\backslash\Phi$ is the fifth solution of \eqref{y-eq1}, then we have $y_0=(\bm y_0^{2^k}+ 1)/(y_0^{2^k}+b)$. Under this condition, we can show that $|\Phi|=4$ cannot occur by proving $y_3=y_4$ in \eqref{Phi1}.

If $y_0$, $y_1$ and $y_2$ are three distinct solutions of \eqref{y-eq1} in $\mu_{2^m+1}$, then by Lemma \ref{lem.eq} one has that
\begin{equation}\label{y3}
y_3=\frac{y_1y_2+\xi y_0 y_2+\xi^2 y_0 y_1}{y_0+\xi y_1+\xi^2 y_2}
\end{equation}
for some $\xi\in\mathbb{F}_{2^2}\backslash\mathbb{F}_{2}$.
Substituting $y_2=(\bm y_1^{2^k}+ 1)/(y_1^{2^k}+b)$ into \eqref{y3} gives
\begin{equation}\label{y3-y0y1}
y_3=\frac{\beta_1 y_1^{2^k+1}+\beta_2 y_1^{2^k}+\beta_3 y_1+\beta_4}{\alpha_1 y_1^{2^k+1}+\alpha_2 y_1^{2^k}+\alpha_3 y_1+\alpha_4},
\end{equation}
where
\begin{eqnarray}
 && \alpha_1=\xi,\,\,\alpha_2=y_0+\bm \xi^2,\,\,\alpha_3=b\xi,\,\,\alpha_4=b y_0+\xi^2, \nonumber\\
&& \beta_1=\xi^2 y_0+\bm ,\,\,\beta_2=\bm \xi y_0,\,\,\beta_3=b\xi^2 y_0+1,\,\,\beta_4=\xi y_0. \label{beta}
\end{eqnarray}
%
A direct calculation by using \eqref{y3-y0y1} gives
\begin{eqnarray}
 \frac{b y_3+1}{y_3+\bm}&=&\frac{(\alpha_1+b \beta_1) y_1^{2^k+1}+(\alpha_2+b \beta_2) y_1^{2^k}
+(\alpha_3+b \beta_3) y_1 +\alpha_4+b \beta_4}{(\bm \alpha_1+ \beta_1) y_1^{2^k+1}+(\bm \alpha_2+ \beta_2) y_1^{2^k}+(\bm\alpha_3+ \beta_3) y_1 +\bm\alpha_4+ \beta_4}, \label{y3-eq1} \\
y_3^{2^k}&=&\frac{\beta_1^{2^k} y_1^{2^{2k}+2^k}+\beta_2^{2^k} y_1^{2^{2k}}+\beta_3^{2^k} y_1^{2^k}+\beta_4^{2^k}}{\alpha_1^{2^k} y_1^{2^{2k}+2^k}+\alpha_2^{2^{k}} y_1^{2^k}+\alpha_3^{2^k} y_1^{2^k}+\alpha_4^{2^k}}. \label{y3-eq2}
\end{eqnarray}
Since $y_1$ is a solution of \eqref{y-eq1}, we have $[(b+\bm^{2^k})y_1+\bm^{2^k+1}+1]y_1^{2^{2k}}=(b^{2^k+1}+1)y_1+b^{2^k}+\bm$.  We then claim that
$(\bm^{2^k}+b)y_1+\bm^{2^k+1}+1\ne0$. Otherwise, we have $y_1=(\bm^{2^k+1}+1)/(\bm^{2^k}+b)$  since $\bm^{2^k}+b\ne0$ due to $\gcd(m+k,2m)=1$ and $b\ne 1$. Using the fact $y_1\in\mu_{2^m+1}$ gives
$$(\bm^{2^k+1}+1)^{2^m+1}+(\bm^{2^k}+b)^{2^m+1}=(b\bm)^{2^k+1}+(b\bm)^{2^k}+b\bm+1=(b\bm+1)^{2^k+1}=0,$$
i.e., $b\bm=1$, which leads to $y_1=\bm(\bm^{2^k}+b)/(\bm^{2^k}+b)=\bm$, a contradiction with $y_1+\bm\ne0$. Thus we have $(\bm^{2^k}+b)y_1+\bm^{2^k+1}+1\ne0$ and then by \eqref{y-eq1} we obtain
$$y_1^{2^{2k}}=\frac{(b^{2^k+1}+1)y_1+b^{2^k}+\bm}{(b+\bm^{2^k})y_1+\bm^{2^k+1}+1}.$$
Plugging it into \eqref{y3-eq2} leads to
\begin{equation}\label{y3-eq2.2}
y_3^{2^k}=\frac{\gamma_1 y_1^{2^k+1}+\gamma_2 y_1^{2^k}
+\gamma_3 y_1 +\gamma_4}{\delta_1 y_1^{2^k+1}+\delta_2 y_1^{2^k}+\delta_3 y_1 +\delta_4},
\end{equation}
where
\begin{eqnarray}
&&\gamma_1=f_1(\beta_1,\,\beta_3),\,\,\gamma_2=f_2(\beta_1,\,\beta_3),\,\,
\gamma_3=f_1(\beta_2,\,\beta_4),\,\,\gamma_4=f_2(\beta_2,\,\beta_4), \label{gamma} \\
&&\delta_1=f_1(\alpha_1,\,\alpha_3),\,\,\delta_2=f_2(\alpha_1,\,\alpha_3),\,\,
\delta_3=f_1(\alpha_2,\,\alpha_4),\,\,\delta_4=f_2(\alpha_2,\,\alpha_4).\nonumber
\end{eqnarray}
Here $f_1(x,\,y):=(b^{2^k+1}+1)x^{2^k}+(b+\bm^{2^k})y^{2^k}$ and $f_2(x,\,y):=(b^{2^k}+\bm)x^{2^k}+(\bm^{2^k+1}+1)y^{2^k}$.
Using the values of $\beta_1,\beta_3$ given by \eqref{beta} and $\gamma_1$ given by \eqref{gamma}, a direct calculation gives
\begin{equation}\label{gamma1}
    \gamma_1  =(b^{2^k+1}+1)(\xi^2 y_0+\bm)^{2^k}+(b+\bm^{2^k})(b\xi^2 y_0+1)^{2^k}
 =((b\bm)^{2^k}+1)(\xi y_0^{2^k}+b).
\end{equation}
The last equality follows from the fact $\xi^{2^k}=\xi^2$ due to $k$ odd since $m$ is even and $\gcd(k,\,m)=1$.
Recall that $y_0^{2^k}=(by_0+1)/(y_0+\bm)$. This together with \eqref{gamma1} yields
$$\gamma_1=\frac{(b\bm)^{2^k}+1}{y_0+\bm}(b\xi^2 y_0+b\bm+\xi)=\frac{(b\bm)^{2^k}+1}{y_0+\bm}(\alpha_1+b\beta_1).$$
Straightforward calculations give $\gamma_i=((b\bm)^{2^k}+1)(\alpha_i+b\beta_i)/(y_0+\bm)$ for $i=2,3,4$ and $\delta_i=((b\bm)^{2^k}+1)(\bm\alpha_i+\beta_i)/(y_0+\bm)$ for $i=1,2,3,4$. Then by \eqref{y3-eq1} and \eqref{y3-eq2.2}, we can deduce
$y_3^{2^k}=(b y_3+1)/(y_3+\bm)$ which implies that $y_3=y_4$, a contradiction with \eqref{Phi1}. Therefore $|\Phi|\leq2$, which shows ${\rm DDT}_F(1,\,b)\leq2$ for $b\in\fn\backslash\mathbb{F}_2$. That is to say, $F(x)$ is locally-APN.

Combining Case 1 and Case 2,  we conclude that ${\rm DDT}_F(1,b)=2^m$ if $b=1$ and ${\rm DDT}_F(1,b)\leq2$ if $b\in\fn\backslash\{1\}$. Consequently, $\omega_{2^m}=1$. Then the differential spectrum of $F(x)$ follows from the well-known identities
$$\omega_0+\omega_2+\omega_{2^m}=2^{2m}, \; 2\omega_2+2^{m}\omega_{2^m}=2^{2m},$$ which completes the proof.
\end{proof}

\begin{remark}
Note that the conjugate of $d$ and the inverse of $d$ (when it exists) are again normalized Niho exponents if $d$ is a normalized Niho exponent. Concretely, let $d=s(2^m-1)+1$ with $2\leq s\leq 2^m$, then the power function $x^{s(2^m-1)+1}$ over $\fn$ for each $s$ in $\{s, 1-s, s/(2s-1), (s-1)/(2s-1)\} ({\rm mod\;} 2^m+1)$ has the same differential spectrum since differential spectrum is invariant under the above two operations.
Computer experiments indicate that locally-APN power functions in Theorem \ref{thm.DS} cover all Niho type locally-APN power functions
for $2\leq m\leq10$.
\end{remark}

\begin{remark}
The power function $F(x)=x^{s(2^m-1)+1}$ over $\fn$ studied in this paper is a permutation if and only if $\gcd(2s-1,2^m+1)=1$, i.e., $\gcd(2^k-1,\,2^m+1)=1$, which indicates that Theorem \ref{thm.DS} produces locally-APN permutations over $\fn$ when $m$ is even and $\gcd(m,k)=1$. This may be of independent interest regarding to the big APN problem \cite{BDMW}.
\end{remark}


\section{The boomerang spectrum of the Niho power function $F(x)=x^{s(2^m-1)+1}$ }\label{boom}

In this section, we discuss the boomerang spectrum of the locally-APN function obtained in the previous section, namely,  $F(x)=x^{s(2^m-1)+1}$ with $s=(2^k+1)^{-1}$ over $\fn$, where $m$ and $k$ are positive integers such that $\gcd(2^k+1,\,2^m+1)=1$ and $\gcd(k,\,m)=1$.

An equivalent definition of the BCT of $F(x)$ was given in \cite{LQSL}, which not only allows us to compute the BCT conveniently without using the compositional inverse $F(x)^{-1}$ but also works for non-permutations.

\begin{lem}{\rm(\cite{LQSL})}\label{BCT}
For a function $F(x)$ over $\mathbb{F}_{2^n}$, the BCT entry of $F(x)$ at point $(a,\,b)\in\mathbb{F}_{2^n}^2$, denoted by ${\rm BCT}_F(a,\,b)$, is the number of solutions $(x,\,y)\in\mathbb{F}_{2^n}^2$ of the system of equations
\begin{eqnarray*}
  F(x+a)+F(y+a)&=&b,\\
  F(x)+F(y)&=& b.
\end{eqnarray*}
\end{lem}

The boomerang spectrum of $F(x)$ can be determined  as below.

\begin{thm}\label{thm.BS}
Let $F(x)=x^{s(2^m-1)+1}$ be a power function over $\fn$, where $m$ and $k$ are positive integers with $\gcd(2^k+1,\,2^m+1)=1$ and $s=(2^k+1)^{-1}$ denotes the multiplicative inverse modulo $2^m+1$. If $\gcd(k,\,m)=1$, then the boomerang spectrum of $F(x)$ is given by
$$\mathbb{BS}_F=\big\{\nu_0=2^{2m-1},\,\nu_2=2^{2m-1}-2^{m},\,\nu_{2^m}=2^{m-1},\,\nu_{2^m+2}=2^{m-1}-1\big\}$$
if $m$ is odd; and otherwise
$$\mathbb{BS}_F=\big\{\nu_0=2^{2m-1},\,\nu_2=2^{2m-1}-2^{m},\,\nu_{2^m}=2^{m-1}-1,\,\nu_{2^m+2}=2^{m-1}\big\}.$$
\end{thm}

\begin{proof}
According to Lemma \ref{BCT}, the value of ${\rm BCT}_F(1,\,b)$ is the number of solutions $(x,\,y)\in\fn^2$ of the following system of equations
\begin{eqnarray*}
  (x+1)^d+(y+1)^d&=&b,\\
  x^d+y^d&=& b
\end{eqnarray*}
as $b$ runs through $\fn^*$. Note that it is equivalent to
\begin{eqnarray}
  (x+1)^d+(y+1)^d&=&b,\label{BS1}\\
  \Delta(x)+\Delta(y) &=&0 ,\label{BS2}
\end{eqnarray}
where $\Delta(x)=(x+1)^d+x^d$.

Let $\Delta(x)=\Delta(y)=c$ for some $c\in\fn$. From Theorem \ref{thm.DS}, $\Delta(x)=c$ (resp. $\Delta(y)=c$) has either $0,\,2$ or $2^m$ solutions. Define
$$\Omega_i:=\{c\in\fn: |\{x\in\fn: \Delta(x)=c\}|=i\}$$
for $i=0,\,2,\,2^m$. Notice that $|\Omega_i|=\omega_i$, where $\omega_i$ is given by Theorem \ref{thm.DS}. Then, for $c\in\fn$, we consider the following three cases:

\textbf{Case 1}: $c\in\Omega_0$. That is, $\Delta(x)=c$ has no solution in $\fn$. Hence the system of equations \eqref{BS1}-\eqref{BS2} has no solution for any $b\in \fn^*$.

\textbf{Case 2}: $c\in\Omega_2$. Assume that $\Delta(x)=c$ has exactly two solutions $x_{0}$ and $x_{0}+1$ in $\fn$ for a fixed $c\in \Omega_{2}$. Then in this case \eqref{BS2} has four solutions $(x,y)=(x_0,x_0)$, $(x_0,x_0+1)$, $(x_0+1,x_0)$ and $(x_0+1,x_0+1)$. Then, we can conclude that the system of \eqref{BS1}-\eqref{BS2} has exactly two solutions $(x_{0},x_{0}+1)$ and $(x_{0}+1, x_{0})$ if $b=c\in \Omega_{2}$ and has no solution otherwise.

\textbf{Case 3}: $c\in\Omega_{2^m}$. This case happens if $c=1$. According to Case 1 of the proof of Theorem \ref{thm.DS}, $\Delta(x)=1$ has $2^m$ solutions with the solution set $\fm$. Then $\Delta(x)=\Delta(y)=c=1$ has $2^{2m}$ solutions $(x,\,y)\in\fm\times\fm$. Note that $x^d=x$ if $x\in\fm$ since $d$ is a normalized Niho exponent. This together with \eqref{BS1} implies that $b=(x+1)^d+(y+1)^d=x+y\in\fm$ which has $2^m$ solutions $(x,y)\in\{(x,x+b):x\in\fm\}$ for each $b\in\fm$. Thus, the system of \eqref{BS1}-\eqref{BS2} has exactly $2^{m}$ solutions if $b\in\fm$ and has no solution otherwise.

Combining Cases 1-3, when $b$ runs through $\fn^*$, we conclude that 1) ${\rm BCT}_F(1,\,b)=2^m+2$ if $b\in\fm^*$ and $b\in\Omega_2$; 2) ${\rm BCT}_F(1,\,b)=2^m$ if $b\in\fm^*$ and $b\notin\Omega_2$; 3) ${\rm BCT}_F(1,\,b)=2$ if $b\notin\fm^*$ and $b\in\Omega_2$; 4) ${\rm BCT}_F(1,\,b)=0$ if $b\notin\fm^*$ and $b\notin\Omega_2$.
Thus $\nu_{2^m+2}=|\fm^*\cap\Omega_2|$, $\nu_{2^m}=2^m-1-|\fm^*\cap\Omega_2|$, $\nu_{2}=|\Omega_2|-|\fm^*\cap\Omega_2|$
and $\nu_0=2^{2m}-1-(\nu_{2^m+2}+\nu_{2^m}+\nu_{2})$. Note that $|\Omega_2|=\omega_2=2^{2m-1}-2^{m-1}$.
Therefore, to complete the proof, it suffices to calculate the value of $|\fm^*\cap\Omega_2|$.

According to the definitions of ${\rm DDT}_F(1,\,b)$ and $\Omega_2$ and the fact that $1\notin\Omega_2$, we have
$$|\fm^*\cap\Omega_2|=|\{b\in\fm\backslash\mathbb{F}_2: {\rm DDT}_F(1,\,b)=2\}|.$$
Recall from Case 2 of the proof of Theorem \ref{thm.DS} that ${\rm DDT}_F(1,\,b)\leq|\Phi|\leq2$ when $b\in\fn\backslash\{1\}$, where $\Phi$ is defined by \eqref{eq-phi}.
Next we claim that ${\rm DDT}_F(1,\,b)=2$ if and only if $|\Phi|=2$ when $b\in\fm\backslash\mathbb{F}_2$. Obviously, we have $|\Phi|=2$ if ${\rm DDT}_F(1,\,b)=2$. Now suppose that $|\Phi|=2$. As we proved before, $z=(\bm y^{2^k}+ 1)/(y^{2^k}+b)\in\mu_{2^m+1}$ is also in $\Phi$ if $y\in\Phi$.
Thus, when $b\in\fm\backslash\mathbb{F}_2$, we can assume that
\[\Phi=\{y, z=\frac{b y^{2^k}+ 1}{y^{2^k}+b}\}.\]
Again by the fact $(\bm y^{2^k}+ 1)/(y^{2^k}+b)$ is also in $\Phi$ if $y\in\Phi$, for $z\in\Phi$, we have $(b z^{2^k}+ 1)/(z^{2^k}+b)\in\Phi$ which implies that $y=(b z^{2^k}+ 1)/(z^{2^k}+b)$ due to $b\ne 1$ and $|\Phi|=2$. That is, we have
\[y=\frac{b z^{2^k}+ 1}{z^{2^k}+b}, z=\frac{b y^{2^k}+ 1}{y^{2^k}+b},y^{2^k}=\frac{bz+1}{z+b}, z^{2^k}=\frac{by+1}{y+b}.\]
Let $x=(bz+1)y/(y+z)$. Note that $b\in\fm$ and $y,z\in\mu_{2^m+1}$. Then we have
\begin{eqnarray*}
  x^{2^m-1}=\Big(\frac{(bz+1)y}{y+z}\Big)^{2^m-1}=\frac{z+b}{(bz+1)y}=\frac{1}{y^{2^k+1}}=y^{-s^{-1}},\\ 
(x+1)^{2^m-1}=\Big(\frac{(by+1)z}{y+z}\Big)^{2^m-1}=\frac{y+b}{(by+1)z}=\frac{1}{z^{2^k+1}}=z^{-s^{-1}}. 
\end{eqnarray*}
This leads to
$$(x+1)^{s(2^m-1)+1}+x^{s(2^m-1)+1}=\frac{x+1}{z}+\frac{x}{y}=\frac{by+1}{y+z}+\frac{bz+1}{y+z}=b,$$
i.e., $x=(bz+1)y/(y+z)$ is a solution of \eqref{D(F)}, which implies that ${\rm DDT}_F(1,\,b)>0$. Then, by ${\rm DDT}_F(1,\,b)$ is even and ${\rm DDT}_F(1,\,b)\le 2$, we have ${\rm DDT}_F(1,\,b)=2$. This shows that ${\rm DDT}_F(1,\,b)=2$ if and only if $|\Phi|=2$ when $b\in\fm\backslash\mathbb{F}_2$. Hence, we obtain
 $$|\fm^*\cap\Omega_2|=|\{b\in\fm\backslash\mathbb{F}_2: {\rm DDT}_F(1,\,b)=2\}|=|\{b\in\fm\backslash\mathbb{F}_2: |\Phi|=2\}|.$$

Observe that $y^{2^k}+b\ne0$ for $y\in\mu_{2^m+1}\backslash\{1\}$ and $b\in\fm\backslash\mathbb{F}_2$. Then from \eqref{eq-phi}, one knows that $|\Phi|=2$ if and only if there are two $y\in\mu_{2^m+1}\backslash\{1\}$ such that
\begin{equation}\label{sys.binFm}
\begin{aligned}
(b^{2^k}+b)y^{2^{2k}+1}+(b^{2^k+1}+1) y^{2^{2k}}+ (b^{2^k+1}+1)y+b^{2^k}+b&= 0, \\
  y^{2^k+1}+b y^{2^k}+by+1&\ne 0.
\end{aligned}
\end{equation}
Let $y=(\tau+\overline{\theta})/(\tau+\theta)$ for a fixed $\theta\in\fn\backslash\fm$ and $\tau\in\fm$. By Lemma \ref{lem.u},  \eqref{sys.binFm} becomes
\begin{eqnarray}
  (b+1)^{2^k+1}(\theta+\ta)\tau^{2^{2k}}+(b+1)^{2^k+1}(\theta+\ta)^{2^{2k}} \tau
  + \varepsilon_1&=& 0, \label{ta-eq1}\\
  (b+1)(\theta+\ta)\tau^{2^k}+(b+1)(\theta+\ta)^{2^k} \tau+\varepsilon_2&\ne& 0 \label{ta-eq2},
\end{eqnarray}
where
\begin{eqnarray*}
\varepsilon_1&=&(b^{2^k}+b)(\theta+\ta)^{2^{2k}+1}
+(b+1)^{2^k+1}(\theta^{2^{2k}}\ta+\theta\ta^{2^{2k}}),\\
\varepsilon_2&=&(\theta+\ta)^{2^{k}+1}
+(b+1)(\theta^{2^{k}}\ta+\theta\ta^{2^{k}}).
\end{eqnarray*}
Substituting $\tau$ with $(\theta+\ta)\tau$ and dividing both sides of \eqref{ta-eq1} and \eqref{ta-eq2} by $(b+1)^{2^k+1}(\theta+\ta)^{2^{2k}+1}$ and $(b+1)(\theta+\ta)^{2^{k}+1}$ respectively give
\begin{eqnarray}
  \tau^{2^{2k}}+ \tau+ \lambda_1&=& 0, \label{ta-eq3}\\
  \tau^{2^k}+ \tau+\lambda_2&\ne& 0 \label{ta-eq4},
\end{eqnarray}
where
\begin{equation}\label{lambda2}
\lambda_1=\frac{b^{2^k}+b}{(b+1)^{2^k+1}}+\frac{\theta^{2^{2k}}\ta
+\theta\ta^{2^{2k}}}{(\theta+\ta)^{2^{2k}+1}},\,\,\,
\lambda_2=\frac{1}{b+1}+\frac{\theta^{2^{k}}\ta+\theta\ta^{2^{k}}}{(\theta+\ta)^{2^{k}+1}}.
\end{equation}
Let $T=\tau^{2^k}+ \tau+\lambda_2$, then $T\in\fm$ since $\lambda_2\in\fm$. Using $\lambda_1=\lambda_2^{2^k}+\lambda_2$, \eqref{ta-eq3} becomes $T^{2^k}+T=0$, which leads to $T=1$ due to $\gcd(k,\,m)=1$ and \eqref{ta-eq4}. That is, we have
\begin{equation}\label{ta-eq5}
\tau^{2^k}+ \tau+\lambda_2+1=0.
\end{equation}
Lemma \ref{lem.eq1} states that \eqref{ta-eq5} has two solutions in $\fm$ if and only if
${\rm Tr}_1^m(\lambda_2+1)=0$, i.e.,
\begin{equation}\label{tr(1/(b+1))}
{\rm Tr}_1^m\Big(\frac{1}{b+1}\Big)={\rm Tr}_1^m\bigg(\frac{\theta^{2^{k}+1}+\ta^{2^{k}+1}}{(\theta+\ta)^{2^{k}+1}}\bigg)
\end{equation}
by using $\lambda_2$ in \eqref{lambda2}.
Next we claim that ${\rm Tr}_1^m((\theta^{2^{k}+1}+\ta^{2^{k}+1})/(\theta+\ta)^{2^{k}+1})=1$.
By Lemma \ref{lem.theta}, one gets
$$\theta^{2^{k}+1}+\ta^{2^{k}+1}=(\theta+\ta)^{2^k+1}+\sum_{i=0}^{k-1}(\theta\ta)^{2^i}
(\theta+\ta)^{2^k-2^{i+1}+1}.$$
Then we have
$$\begin{aligned}
  {\rm Tr}_1^m\bigg(\frac{\theta^{2^{k}+1}+\ta^{2^{k}+1}}{(\theta+\ta)^{2^{k}+1}}\bigg) &=
  {\rm Tr}_1^m(1)+ {\rm Tr}_1^m\bigg(\sum_{i=0}^{k-1}\Big(\frac{\theta\ta}{(\theta+\ta)^2}\Big)^{2^i}\bigg)\\
   &=m+k{\rm Tr}_1^m\bigg(\frac{\theta\ta}{(\theta+\ta)^2}\bigg)
   \\&=m+k=1,
\end{aligned}$$
where the third equality holds due to
$${\rm Tr}_1^m\bigg(\frac{\theta\ta}{(\theta+\ta)^2}\bigg)={\rm Tr}_1^m\bigg(\frac{1}{\theta/\ta+\ta/\theta}\bigg)=1$$
by using Lemma \ref{lem.eq.x^2} with the fact that $\theta/\ta $ and $\ta/\theta$ are two solutions of $x^2+(\theta/\ta+\ta/\theta)x+1=0$ in $\mu_{2^m+1}\backslash\{1\}$.
Moreover, the last equality follows from $m$ and $k$ have different parity.
Therefore, when $b\in\fm^*\cap\Omega_2$, we have $|\Phi|=2$ if and only if ${\rm Tr}_1^m(1/(b+1))=1$ by \eqref{tr(1/(b+1))}. Consequently, $|\fm^*\cap\Omega_2|=2^{m-1}$ if $m$ is even and otherwise $|\fm^*\cap\Omega_2|=2^{m-1}-1$. This completes the proof.
\end{proof}

\begin{remark}
Note that $s=2$ if one takes $k=m-1$ in Theorems \ref{thm.DS} and \ref{thm.BS}. The differential spectrum and boomerang spectrum of the Niho type power function for $s=2$ have been determined by Blondeau et al. \cite{BCC1} and Yan et al. \cite{YZZ}, respectively.
\end{remark}

\begin{remark}
Experimental results indicate that the power function $F(x)=x^{s(2^m-1)+1}$ is not locally-APN when $\gcd(2^k+1,\,2^m+1)=1$ and $\gcd(k,\,m)>1$.
In this case, the differential spectrum and boomerang spectrum of $F(x)$ can be studied in a similar way with more efforts in discussing the involved equations over finite fields.
\end{remark}

\section{Conclusion}\label{conc}
Twelve years ago, Blondeau, Canteaut, and Charpin introduced locally-APN functions in the context of the block cipher in symmetric cryptography.
 Only a few studies on these functions have been developed, and few infinite families of such functions have been exhibited.
 In this paper, we have explored the corpus of locally-APN functions over the finite field $\fn$ by studying specifically the locally-APN-ness property of the Niho type power function $F(x)=x^{s(2^m-1)+1}$ with $s=(2^k+1)^{-1}$ over $\fn$ and completely determined its differential spectrum, where $m$ and $k$ are positive integers such that $\gcd(k,\,m)=1$, $\gcd(2^k+1,2^m+1)=1$ and $(2^k+1)^{-1}$ denotes the multiplicative inverse modulo $2^m+1$. Further, we completely determined the boomerang spectrum of $F(x)$ based on its differential spectrum, which generalizes a recent result presented by Yan et al. \cite{YZZ}. Our experimental results for $2\leq m\leq10$ indicate that all Niho type locally-APN power functions over $\fn$ are covered by Theorem \ref{thm.DS}, which accentuates the relevance of this result. It is, therefore, interesting to confirm whether this is true for a general $m$. Moreover, since the notion of locally-APN-ness can be viewed as a relaxation of the one of APN-ness, it offers the possibility to discover new avenues in designing S-boxes for cryptographic uses in cryptosystems involving block ciphers.

\section*{Acknowledgments}

This work was supported by the National Key Research and Development Program of China (No. 2021YFA1000600), the National Natural Science Foundation of China (No. 62072162), the Natural Science Foundation of Hubei Province of China (No. 2021CFA079) and the Knowledge Innovation Program of Wuhan-Basic Research (No. 2022010801010319).

\end{document}